\documentclass[aip,jcp,reprint,numerical,superscriptaddress]{revtex4-1}
\usepackage{amsmath,amssymb,physics,booktabs,graphicx,verbatim,float,csquotes}
\usepackage{epsfig}
\usepackage[breaklinks=true,colorlinks,citecolor=blue,linkcolor=blue,urlcolor=blue]{hyperref}
\usepackage[caption=false]{subfig}
\usepackage[T1]{fontenc}
\usepackage{ae,aecompl}
\usepackage{times}
\usepackage{pslatex}
\usepackage[inline]{enumitem}
\usepackage[framed]{ntheorem}
\theorembodyfont{\upshape}
\theoremstyle{plain}
\newtheorem*{theorem}{Theorem:}

\newtheorem*{proof}{Proof:}
\newtheorem{rel}{Relation}
\newcommand{\vect}[1]{\textbf{#1}}

\usepackage{relsize}
\begin{document}
\title{A general penalty method for density-to-potential inversion}
\author{Ashish Kumar}  \email[]{ashishkr@iitk.ac.in}
\affiliation{Department of Physics, Indian Institute of Technology Kanpur, Kanpur-208016, India}
\author{Manoj K. Harbola}  \email[]{mkh@iitk.ac.in}
\affiliation{Department of Physics, Indian Institute of Technology Kanpur, Kanpur-208016, India}
\date{\today}
\begin{abstract}
A general penalty method is presented  for the construction of Kohn-Sham system for given density through Levy's constrained-search. The method uses a functional $S[\rho]$ of one's choice. Different forms of $S[\rho]$ are  employed  to calculate the kinetic energy and exchange-correlation potential of atoms, jellium spheres, and the Hookium and  consistency among results obtained from them is shown for each system.
\end{abstract}
\maketitle
\section{Introduction}\label{sec5.1}
With increasing accuracy of density functional theory (DFT) \cite{Hohn, Kohn_1965,Becke_1988, Perdew_PRL.77.3865,Jianmin_PRL.91.146401, Sun_PRL.115.036402, Patra_2019} calculations, it is imperative that exact results \cite{Parr_jcp_1978,PPLB,LPS_1984,UvBarth_1985,Levy_Perdew_1985} also be made available wherever possible. This is important both from a fundamental point of view as well as for gaining insights into the working of functionals employed to perform DFT calculations. Majority of DFT calculations are performed within its Kohn-Sham (KSDFT) formalism \cite{Kohn_1965}. The key ingredient of a KSDFT calculation is the exchange-correlation energy functional $E_{xc}[\rho]$, where $\rho(\vect{r})$ is a given density, that incorporates all many-body effects in it; the functional derivative of $E_{xc}[\rho]$ with respect to the density gives the exchange-correlation potential $v_{xc}(\vect{r}) = \frac{\delta E_{xc}[\rho]}{\delta \rho(\vect{r})}$ used in the Kohn-Sham equation (atomic units used throughout the paper)
\begin{equation}\label{eq5.1.1}
\Big[ -\frac{1}{2} \nabla^2 + v_{ext}(\vect{r}) +v_H(\vect{r}) +v_{xc}(\vect{r})\Big]\phi_i(\vect{r}) = \epsilon_i\phi_i(\vect{r}).
\end{equation}
Here $ v_{ext}(\vect{r}) +v_H(\vect{r}) +v_{xc}(\vect{r})$ is known as the Kohn-Sham potential. Orbitals $\{ \phi_i(\vect{r})\}$ obtained from the Kohn-Sham equation give the ground state density
\begin{equation}\label{eq5.1.2}
\rho(\vect{r}) = \sum_{i=1}^{N} f_i| \phi_i(\vect{r})|^2,
\end{equation}
where $\{f_i\}$ are the electron occupation number of orbitals $\{i\}$ and $N$ denotes the uppermost filled orbital. In the equation above $v_{ext}(\vect{r})$ is the external potential in which electrons move, and
\begin{equation}\label{eq5.1.3}
v_H(\vect{r}) = \int \frac{\rho(\vect{r}')}{|\vect{r}-\vect{r}'|}d\vect{r}'
\end{equation}
is the Hartree potential.  The total energy of the system is given as
\begin{equation}\label{eq5.1.4}
\begin{split}
E[\rho] = \sum_{i=1}^{N} f_i \langle \phi_i| -\frac{1}{2}\nabla^2 |\phi_i \rangle
&+ \int v_{ext}(\vect{r}) \rho(\vect{r})d \vect{r} \\
&+ \frac{1}{2}\iint \frac{\rho(\vect{r})\rho(\vect{r}')}{|\vect{r}-\vect{r}'|}d\vect{r}'d\vect{r}' +E_{xc}[\rho].
\end{split}
\end{equation}
As is clear from the description above, in carrying out  Kohn-Sham DFT calculations, $E_{xc}[\rho]$ functional,  and therefore its derivative $v_{xc}[\rho](\vect{r})$, are the only ingredients that are approximated; everything else is known exactly. Therefore to get the exact $v_{xc}(\vect{r})$  for a given density, some other method has to be employed. As noted earlier, knowing $v_{xc}(\vect{r})$ exactly is of interest by itself and is also important to provide insights \cite{Gritsenko_1996, Ruzsinszky_jcp_2006,Teal2,Makmal_2011, Wagner_2012,2014_Gould,Godby_PRA_2016, Rabeet_2017,Hodgson_interatomic_2017,Staroverov_PNAS_18, kummel_18, 2019_Gould,kumar2019accurate} into the working of approximate exchange-correlation functionals $E_{xc}[\rho]$. To do this several inversion schemes  have been proposed \cite{Werden, Stott_1988, Gorling_1992,Wang_1993, Zhao_1993, Zhao_1994, Wang_1993, Vlb_1994,Schipper1997, WY2,Peirs_2003, Stott_2004,Hollins_2017,Wasserman_2017,Finzel2018,Kumar_2019,kanungo_zimmerman_gavini_2019,Gidopoulos_jcp_2020}.  Most of them employ optimization approaches based on fundamental principles of DFT \cite{ Hohn, Levy_1979, Lieb_1983}. These approaches either minimize the non-interacting kinetic energy $T_S[\rho] = \sum_{i=1}^{N} f_i \langle \phi_i| -\frac{1}{2}\nabla^2 |\phi_i \rangle$ of electrons with the constraint that the corresponding orbitals lead to the given density $\rho_0(\vect{r})$ \cite{Levy_1979, Zhao_1993, Zhao_1994} or maximize the functional given in Eq. (\ref{eq5.1.6}) by varying the Kohn-Sham potential \cite{Lieb_1983,WY2,Kumar_2019}. The latter approach does not involve imposition of any constraint.
\par Minimization of the kinetic energy with constraint  can be carried out as minimization of the functional
\begin{equation}\label{eq5.1.5}
J_{\rho_0,v}[\rho] = T_S[\rho] + \int v(\vect{r}) (\rho(\vect{r})-\rho_0(\vect{r}))d \vect{r},
\end{equation}
with respect to $ \{ \phi_i(\vect{r})\}$, where  $\rho(\vect{r})$ is given by Eq. (\ref{eq5.1.2})  and $v(\vect{r})$ is the Lagrange multiplier to enforce the condition $\rho(\vect{r}) = \rho_0(\vect{r})$.  This leads to  the equation
\begin{equation}\label{eq5.1.5b}
\Big[ -\frac{1}{2} \nabla^2 + v(\vect{r}) \Big]\phi_i(\vect{r}) = \epsilon_i\phi_i(\vect{r}),
\end{equation}
for  $ \{ \phi_i(\vect{r})\}$ where $ \{\epsilon_i \}$ are  the  Lagrange multipliers corresponding to the orbital normalization. This shows that Lagrange multiplier function $v(\vect{r})$, which enforces density constraint, is the Kohn-Sham potential $v_{ext}(\vect{r}) +v_H(\vect{r}) +v_{xc}(\vect{r})$  (see Eq. (\ref{Fig5.5.1})) and leads to the exchange-correlation potential for the density $\rho_0(\vect{r})$.
\par A practical approach to carry out minimization of $ T_S[\rho]$ is the penalty method. As we will see later, in this method $v(\vect{r})$ is expressed as a function of $\rho(\vect{r})$ so that $v(\vect{r}) = v[\rho](\vect{r})$, $J_{\rho_0,v}[\rho] = J_{\rho_0}[\rho]$ ,   and $\rho(\vect{r})$ is varied to minimze $J_{\rho_0}[\rho]$, which also gaurantees that $\rho(\vect{r}) = \rho_0(\vect{r})$.  This is explained in detail in  Appendix (\ref{sec5a1}) . 
\par In the second approach, unconstrained maximization of the functional
\begin{equation}\label{eq5.1.6}
J_{\rho_0}[v]  = \sum_{i=1}^{N} f_i \langle \phi_i[v] | -\frac{1}{2}\nabla^2 |\phi_i [v] \rangle
+ \int v(\vect{r}) (\rho [v](\vect{r})-\rho_0(\vect{r}))d \vect{r}
\end{equation}
with respect to $v(\vect{r})$ is carried out, with $\phi_i (\vect{r})$ being the solution of Eq. (\ref{eq5.1.5b}). In this method  Eq. (\ref{eq5.1.5b}) is solved and $v(\vect{r})$ is varied until the quantity of Eq. (\ref{eq5.1.6}) attains its maximum. Important connection between the methods described by Eq. (\ref{eq5.1.5}) and Eq. (\ref{eq5.1.6}) is that the  same functional is employed to search for $v(\vect{r})$. However in Eq. (\ref{eq5.1.5}) $v(\vect{r})$ is expressed as a functional of the density and  minimization is carried out with respect to $\rho(\vect{r})$ whereas in Eq. (\ref{eq5.1.6}), the density is expressed in terms of  $v(\vect{r})$ and  maximization is carried out with respect to $v(\vect{r})$. While minimization of Eq. (\ref{eq5.1.5}) is done using the Zhao-Morrison-Parr (ZMP) method (described below) in  majority of cases,  several different approaches have been proposed for maximization of the functional in  Eq. (\ref{eq5.1.6})  .
\par The purpose of this work is to provide a general connection between the minimization method of Eq. (\ref{eq5.1.5}) and different methods proposed for performing unconstrained maximization of the functional $J_{\rho_0}[v]$ of Eq. (\ref{eq5.1.6}). We show that for  each one of the latter methods, there is a corresponding method employing Eq. (\ref{eq5.1.5}).  This connection makes use of a general condition \cite{Kumar_2019} derived   recently, and given in Eq. (\ref{eq5.4.2}) later, in the context of showing the universal nature of different methods of density-to-potential inversion.
\par In the following we first outline the definition of the universal functional (see Eq. (\ref{eq5.2.2}) below) of density functional theory and discuss how this is used to obtain the exchange-correlation potential of Kohn-Sham density functional theory through Eq. (\ref{eq5.1.5}) and Eq. (\ref{eq5.1.6}). In particular we describe the ZMP method for minimizing the Kohn-Sham kinetic energy with constraint and several methods and their universality for maximizing the functional of Eq. (\ref{eq5.1.6}). Based on the latter, we then show that the ZMP method also has a general nature and several functionals other than those proposed by  ZMP can be equally effective for obtaining the Kohn-Sham exchange-correlation potential. This then brings forth and demonstrates the conjugate relationship (see section \ref{subsec_cojugate} below) between density $\rho (\vect{r})$ and external potential $v_{ext}(\vect{r})$ used in Levy's \cite{ Levy_1979} and Lieb's \cite{Lieb_1983} definition of the universal functional of DFT. In addition, it provides several options for the functional of ones choice  to be used in Eq. (\ref{eq5.1.5}) to calculate the exchange-correlation potential by minimizing the non-interacting Kinetic enenrgy functional.
\section{Universal functional $F[\rho]$ of DFT and generating exchange-correlation potential for a given density}\label{sec5.2}
In DFT , the ground-state energy $E[\rho]$ as a functional of ground state density $\rho(\vect{r})$ is written as
\begin{equation}\label{eq5.2.1}
E[\rho] =\int v_{ext}(\vect{r})\rho(\vect{r})d\vect{r} + F[\rho] ,
\end{equation}
where $v_{ext}(\vect{r})$ is the external potential electrons are moving in and $F[\rho]$ is a universal functional of the density. For a given density, this functional is given as \cite{Levy_1979}
\begin{equation}\label{eq5.2.2}
F[\rho]  = \underset{\Psi \to \rho}{min}\langle \Psi |T +V_{ee}| \Psi\rangle,
\end{equation}
where $T$ and $V_{ee}$ are the kinetic and the electron-electron interaction enenrgy operators, respectively. Here the minimization is done over all those wavefunctions $\Psi$  that are  antisymmetric with respect to exchange of electron coordinates and give density $\rho(\vect{r})$.  Hence the definition given by Eq. (\ref{eq5.2.2}) and corresponding search of $\Psi$ is known as the constrained-search approach. For the corresponding Kohn-Sham system, the universal functional is
\begin{equation}\label{eq5.2.3}
F_{KS}[\rho]  = \underset{\Phi \to \rho}{min} \langle \Phi |T| \Phi\rangle,
\end{equation}
where now constrained search is over Slater determinants $\Phi$ made of $N$ orbitals $\{ \phi_i\}$ for $N$ electrons that give the corresponding $\rho(\vect{r})$. The constraint that  $\sum_{i=1}^{N} f_i |\phi_i(\vect{r})|^2 =  \rho(\vect{r})$ can be implemented through penalty method given in Appendix (\ref{sec5a1}). Using this approach ZMP formulated practical scheme for calculation of exchange-correlation potential as described in the  following.
\subsection{Zhao-Morrison-Parr scheme} \label{subsec5.2.1}
Zhao-Morrison-Parr proposed \cite{Zhao_1993, Zhao_1994} that the condition that $\Phi$ lead to the given density $\rho_0(\vect{r})$ can be implemented by demanding that
\begin{equation}\label{eq5.2.1.1}
\frac{1}{2} \iint \frac{ \{\rho (\vect{r})  -\rho_0(\vect{r}) \}\{\rho (\vect{r}') -\rho_0(\vect{r}')\}}{|\vect{r}-\vect{r}'|}d\vect{r}d\vect{r}' = 0.
\end{equation}
Here  $ \rho(\vect{r})  =\sum_{i=1}^{N} f_i |\phi_i(\vect{r})|^2 $. Note that the condition above implies \cite{jackson_ED,Griffiths_ED} that $\rho(\vect{r}) =\rho_0(\vect{r})$. To get the KS exchange-correlation potential , one performs unconstrained minimization of the functional
\begin{equation}\label{eq5.2.1.2}
\begin{split}
\sum_{i=1}^{N} f_i \langle \phi_i| -\frac{1}{2}\nabla^2 |\phi_i \rangle
 + \frac{\lambda}{2}\iint \frac{ \{\rho (\vect{r})  -\rho_0(\vect{r}) \}\{\rho (\vect{r}') -\rho_0(\vect{r}')\}}{|\vect{r}-\vect{r}'|}d\vect{r}d\vect{r}',
\end{split}
\end{equation}
where $\lambda$ is a constant. This leads to the equation
\begin{equation}\label{eq5.2.1.3a}
\begin{split}
\Big[ -\frac{1}{2} \nabla^2 +\lambda\int \frac{\rho(\vect{r}') -\rho_0(\vect{r}')}{|\vect{r}-\vect{r}'|}d\vect{r}'\Big]\phi_i(\vect{r}) = \epsilon_i\phi_i(\vect{r}).
\end{split}
\end{equation}
Now separating out the external potential and the exact Hartree potential for the given density $\rho_0(\vect{r})$, Eq. (\ref{eq5.2.1.3a}) can be written as 
\begin{equation}\label{eq5.2.1.3}
\begin{split}
\Big[ -\frac{1}{2} \nabla^2 + v_{ext}(\vect{r}) +&\int \frac{\rho_0(\vect{r}')}{|\vect{r}-\vect{r}'|}d\vect{r}'   \\
+&\lambda\int \frac{\rho(\vect{r}') -\rho_0(\vect{r}')}{|\vect{r}-\vect{r}'|}d\vect{r}'\Big]\phi_i(\vect{r}) = \epsilon_i\phi_i(\vect{r}).
\end{split}
\end{equation}
The equation above is solved self-consistently and  in the limit of $\lambda \to \infty$ gives the exchange-correlation potential as
\begin{equation}\label{eq5.2.1.4}
v_{xc}(\vect{r})
 =\lim_{\lambda \to ~\infty} ~\lambda \int \frac{ \rho(\vect{r}') -\rho_0(\vect{r}')}{|\vect{r}-\vect{r}'|}d\vect{r}.
\end{equation}
This is known as the Zhao-Morrison-Parr method  and has been implemented \cite{morrison_1995,handy_1996,Jayatilaka_1998,handy_1999,parr_jcp_1999, PRA.69.042512,SAMAL2006217,Leonardo_2014,Rabeet_IJQC_2017} successfully over the years.
The method above is a penalty based method since $\lambda$ is the penalty parameter imposed if  the functional of Eq. (\ref{eq5.2.1.1})  is  non-zero. This is explained in detail in the Appendix. Note that  as required the penalty term 
\begin{equation*}
\iint \frac{\{\rho (\vect{r})  -\rho_0(\vect{r}) \}\{\rho (\vect{r}') -\rho_0(\vect{r}')\}}{|\vect{r}-\vect{r}'|}d\vect{r}d\vect{r}'
\end{equation*}
 is always positive. This is because the term represnts enenrgy of a charge distribution $\{\rho (\vect{r})  -\rho_0(\vect{r})\}$, and that enenrgy is always positive. Mathematically, it can be shown to be positive by using Poisson's equation \cite{Griffiths_ED,jackson_ED} satisfied by the Coulomb potential due to the charge density $ \{\rho (\vect{r})  -\rho_0(\vect{r}) \}$.
\par Having presented a minimization method, we now describe a different method that again uses the universal functional  $F[\rho]$ but in contrast to the ZMP method , it utilizes a maximization scheme.
\subsection{Wu and Yang method}\label{subsec5.2.2}
Zhao-Morrison-Parr method is based on constrained minimization and requires that $\lambda \to \infty$ limit be taken. Wu and Yang \cite{WY2} looked for a method that does not require a constraint condition and proposed that the exchange-correlation potential can be found directly by maximizing the  functional
 \begin{equation}\label{eq5.2.2.1}
J_{\rho_0}[v] =
\sum_{i=1}^{N} f_i \langle \phi_i [v]| -\frac{1}{2}\nabla^2 |\phi_i[v] \rangle
+ \int v(\vect{r}) (\rho[v](\vect{r})-\rho_0(\vect{r}))d \vect{r}
\end{equation}
with respect to $v(\vect{r})$. Here  $\{\phi_i (\vect{r})\}$ are the  solution of the equation
\begin{equation}\label{eq5.2.2.2}
\Big[ -\frac{1}{2} \nabla^2 + v(\vect{r}) \Big]\phi_i(\vect{r}) = \epsilon_i\phi_i(\vect{r}),
\end{equation}
$\rho_0(\vect{r})$ is the given density and $v(\vect{r}) =v_{ext}(\vect{r}) +v_{H}(\vect{r})+v_{xc}(\vect{r})$. In actual  calculations, $v_{ext}(\vect{r})$ and $v_{H}(\vect{r})$ are fixed and therefore $v_{xc}(\vect{r})$ is varied to achieve the maximum. The method has been applied to obtain both the exchange-correlation potential $v_{xc}(\vect{r})$ as well as the external potential \cite {Teal2,Rabeet_IJQC_2017} for scaled electron-electron interaction. Like the ZMP method, this approach too is related to finding the universal functional $F[\rho]$, defined by Lieb as
\begin{equation}\label{eq5.2.2.3}
F[\rho_0] =\underset{v(\vect{r})}{sup} ~~\Big[E[v]
- \int v(\vect{r}) \rho_0(\vect{r})d \vect{r}\Big],
\end{equation}
where $E[v]$ is the energy corresponding to $N$ electrons moving in potential $v(\vect{r})$. For the Kohn-Sham system this reduces to maximizing the functional given in Eq. (\ref{eq5.2.2.1}) with respect to $v(\vect{r})$  or equivalently maximizing
\begin{equation}\label{eq5.2.2.5}
\begin{split}
\sum_{i=1}^{N} f_i \langle \phi_i| -\frac{1}{2}\nabla^2 |\phi_i \rangle
+\int \Big\{ v_{ext}(\vect{r})  &+  \int \frac{\rho_0(\vect{r}')}{|\vect{r}-\vect{r}'|} d\vect{r}' \\ +& v_{xc}(\vect{r}) \Big \}(\rho(\vect{r})-\rho_0(\vect{r}))d \vect{r}
\end{split}
\end{equation}
with respect to $v_{xc}(\vect{r})$. Here $\{\phi_i(\vect{r})\}$ are the solution of the equation
\begin{equation}\label{eq5.2.2.7}
\Big[ -\frac{1}{2} \nabla^2 + v_{ext}(\vect{r}) + \int \frac{\rho_0(\vect{r}')}{|\vect{r}-\vect{r}'|}d\vect{r}'+v_{xc}(\vect{r})\Big]\phi_i(\vect{r}) = \epsilon_i\phi(\vect{r}).
\end{equation}
 To  facilitate calculations further, the Hartree term is used with the Fermi-Amaldi correction \cite{fermi1934orbite}. In many cases where maximization of expression in Eq. (\ref{eq5.2.2.5}) has been carried out, $v_{xc}(\vect{r})$ was expressed as a sum of Gaussian functions.
 \par We now combine the ideas from the  two methods presented above  to give a general penalty method. Before doing that in section \ref{sec5.4}, we first discuss in the following the relationship between the external potential $v_{ext}(\vect{r})$ and the density $\rho(\vect{r})$ in DFT and the associated Legendre transformation between the ground-state energy $E[v_{ext}]$ and the universal functional $F[\rho]$.
\subsection{Conjugate relationship between $\rho (\vect{r})$ and $ v_{ext}(\vect{r}) $ used in Levy's constrained search definition and Lieb's definition of the universal functional $F[\rho]$}\label{subsec_cojugate}
There is a relationship between the two ways that the universal functional has been defined by Eqs. (\ref{eq5.2.2}) and (\ref{eq5.2.2.3}) above.  This arises from the observation \cite{Nalewajski_jcp_1982, Kutzelnigg_JMs_2006} that the energy of a system and the universal functional are related through a Legendre transformation as follows.  Using the Hohenberg-Kohn theorem, the ground-state energy of a system of electrons in an external potential $v_{ext}(\vect{r})$ is given as
\begin{equation}
E[v_{ext}] = \underset{\rho(\vect{r})}{min} \Bigg\{F[\rho] + \int v_{ext}(\vect{r})\rho(\vect{r})d\vect{r} \Bigg\}
\end{equation}
where minimization is done with respect to the density $\rho(\vect{r})$, as indicated.  Notice that in the equation above  $v_{ext}(\vect{r})$ is a given external potential and remains fixed.   This combined with the variational principle for the energy in terms of the wavefunction leads \cite{Levy_1979} to the constrained search definition of the universal functional given by Eq. (\ref{eq5.2.2}). 
\par Now from the first-order perturbation theory
\begin{equation}
\frac{\delta E[v_{ext}]}{\delta v_{ext}(\vect{r})} = \rho(\vect{r})
\end{equation}
which shows that for a given number of electrons, the external potential $v_{ext}(\vect{r})$  and the density $\rho(\vect{r})$ are conjugate variables and we can write the universal functional as a Legendre transform of the energy as 
\begin{equation}
-F[\rho] =  \int v_{ext}(\vect{r})\rho(\vect{r})d\vect{r}  - E[v_{ext}]
\end{equation}
Now suppose the ground-state wavefunction corresponding to $v_{ext}(\vect{r})$  is $\Psi$.  Then we have
\begin{equation}
\begin{split}
F[\rho] &=\langle \Psi |T +V_{ee}| \Psi\rangle \\
& = \langle \Psi |T +V_{ee} +V_{ext}'| \Psi\rangle - \int v_{ext}'(\vect{r})\rho(\vect{r})d\vect{r},
\end{split}
\end{equation}
where $v_{ext}'(\vect{r})$ is  some other external potential.  If the ground-state wavefunction corresponding to $v_{ext}'(\vect{r})$ is $\Psi'$ then by the variational principle for the energy
\begin{equation}
\begin{split}
\langle\Psi |T +V_{ee} +V_{ext}'| \Psi\rangle &\geq \langle\Psi' |T +V_{ee} +V_{ext}'| \Psi'\rangle \\
& =  E[v_{ext}'].
\end{split}
\end{equation}
Thus, for a given density $\rho(\vect{r})$ we have
\begin{equation} \label{cojeq1}
F[\rho] \geq E[v_{ext}'] - \int v_{ext}'(\vect{r})\rho(\vect{r})d\vect{r}   
\end{equation}
Where the equality is satisfied for the true $v_{ext}(\vect{r})$  corresponding to density $\rho(\vect{r})$.  It is evdent that Eq. (\ref{cojeq1}) is equivalent to Eq. (\ref{eq5.2.2.3}).
\par From the discussion above, it is clear that in DFT, the external potential and the ground-state density are conjugate variables and the ground-state energy and the universal functional of the density are related through a Legendre transformation.  As noted earlier, the present work highlights this relationship operationally. 
\begin{figure*}
	\centering
	\vspace {1cm}
	\subfloat[\label{subfig5.5.1a}]{\includegraphics[scale = 0.95]{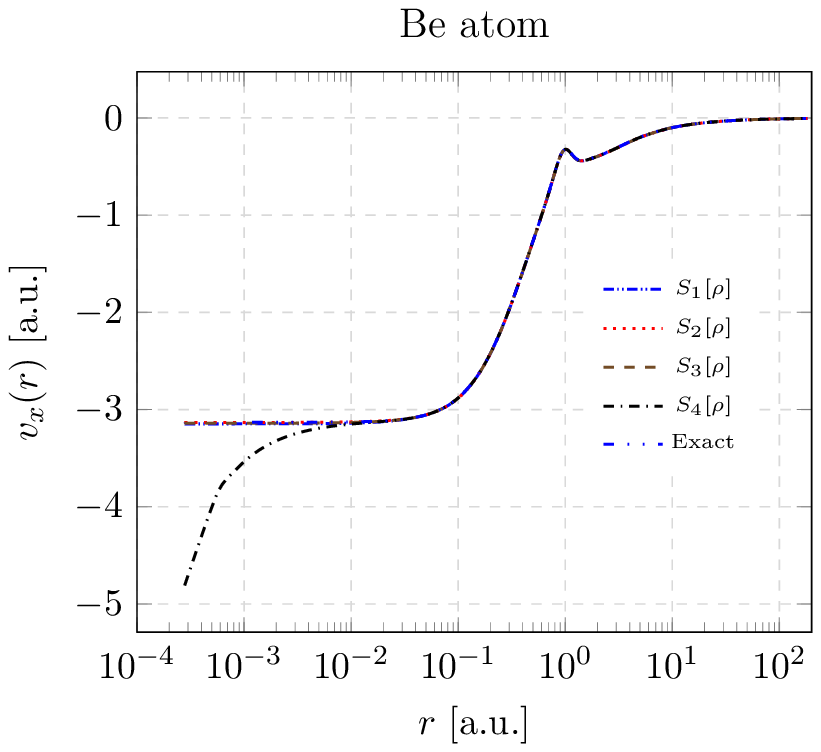}}
	\hfill%
	\subfloat[\label{subfig5.5.1b}]{ \includegraphics[scale = 0.95]{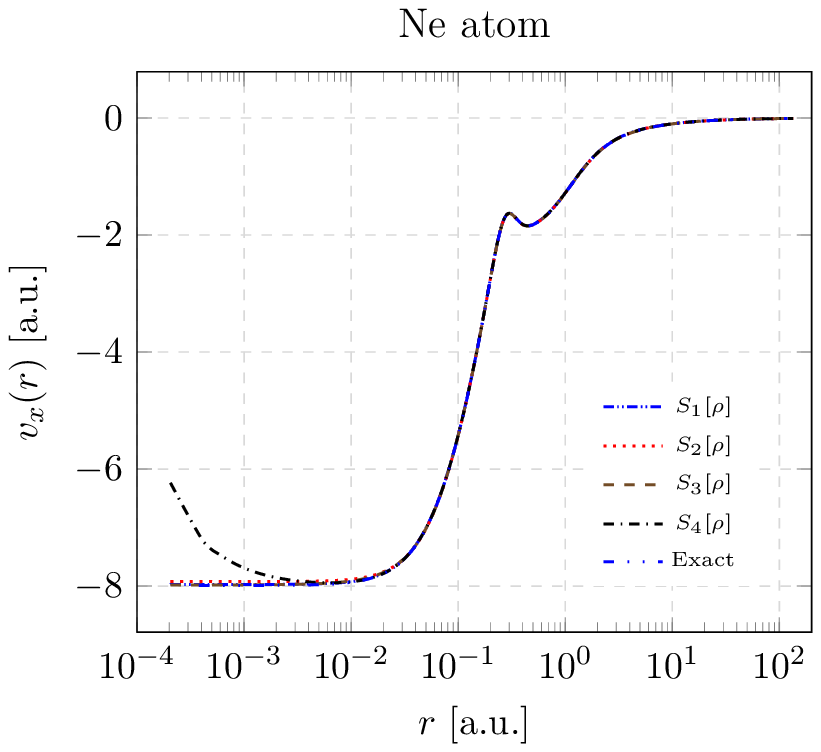}} \\
	\vspace {0.15cm}
	\subfloat[\label{subfig5.5.1c}]{ \includegraphics[scale = 0.95]{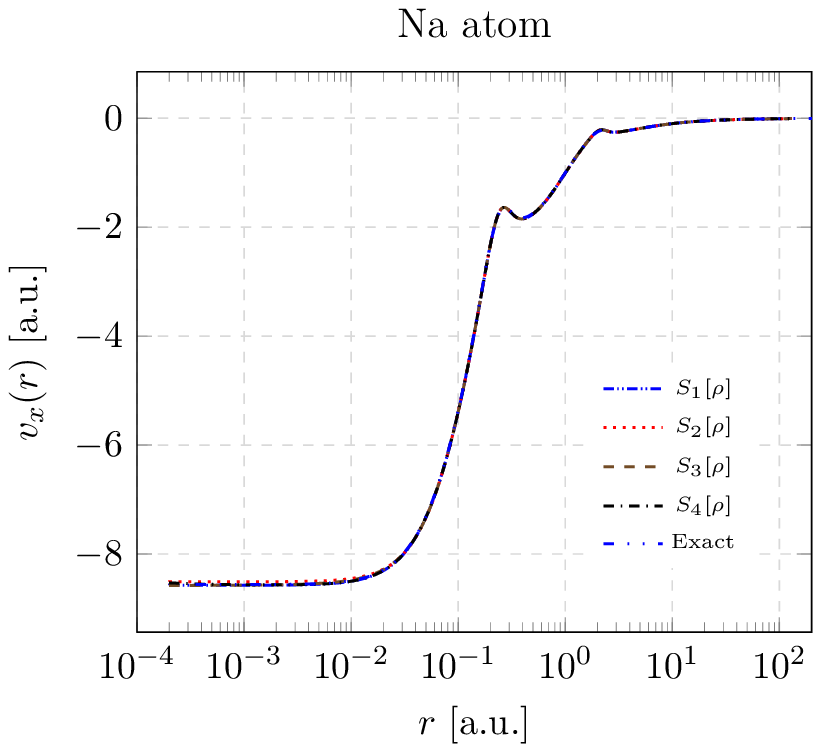}}
	\hfill
	\subfloat[\label{subfig5.5.1d}]{\includegraphics[scale = 0.95]{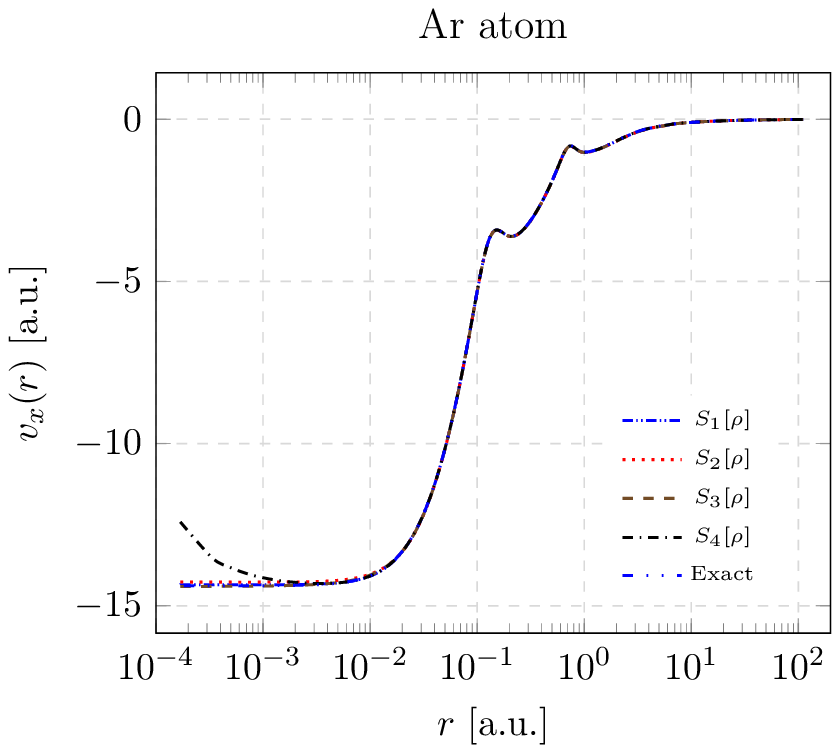}}
	\caption{ \label{Fig5.5.1} Exchange potential for  Hartree-Fock  density of Be, Ne, Na and Ar atoms using  functionals $S_1[\rho]$, $S_2[\rho]$, $S_3[\rho]$ and $S_4[\rho]$ of Eqs. (\ref{eq5.5.1}- \ref{eq5.5.4}).}
\end{figure*}
\begin{figure*}
	\centering
	\vspace {1cm}
	\subfloat[\label{subfig5.5.2a}]{\includegraphics[scale = 0.95]{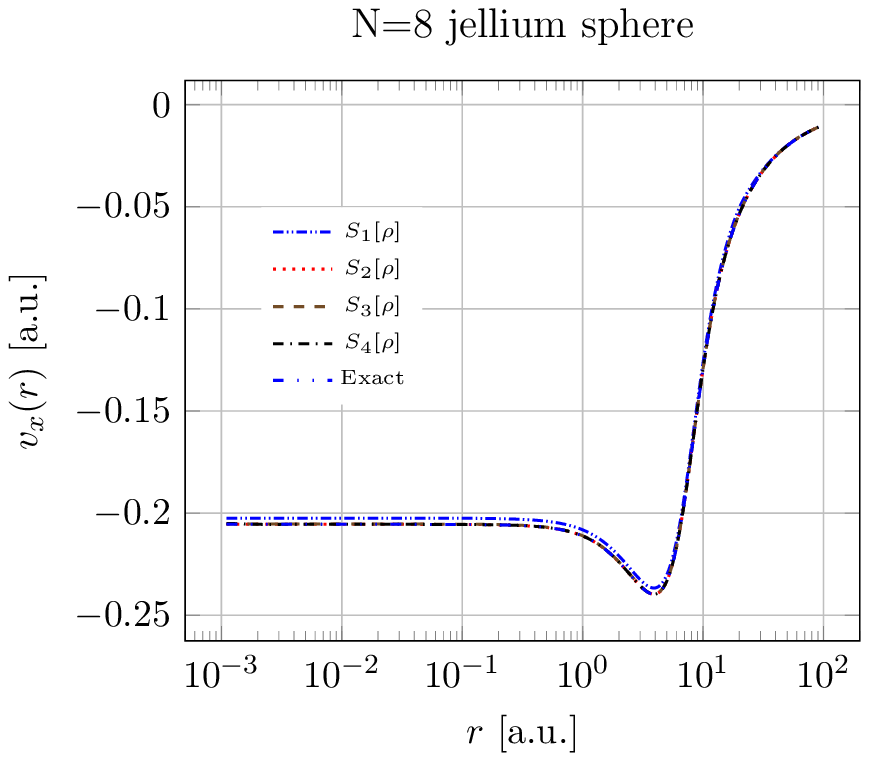}}
	\hfill%
	\subfloat[\label{subfig5.5.2b}]{ \includegraphics[scale = 0.95]{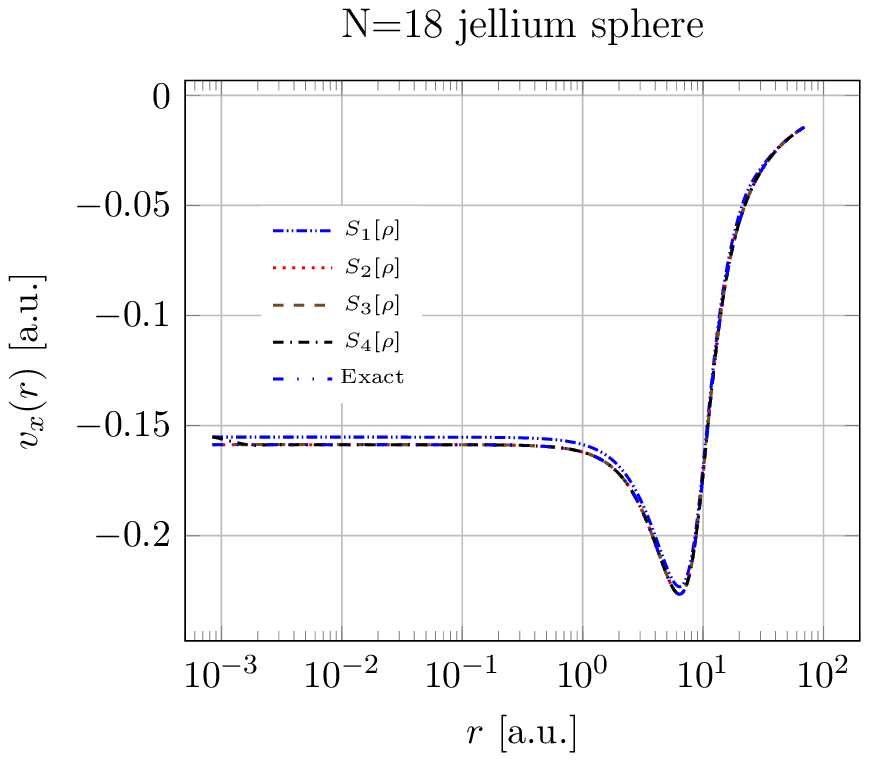}} \\
	\vspace {0.15cm}
	\subfloat[\label{subfig5.5.2c}]{ \includegraphics[scale = 0.95]{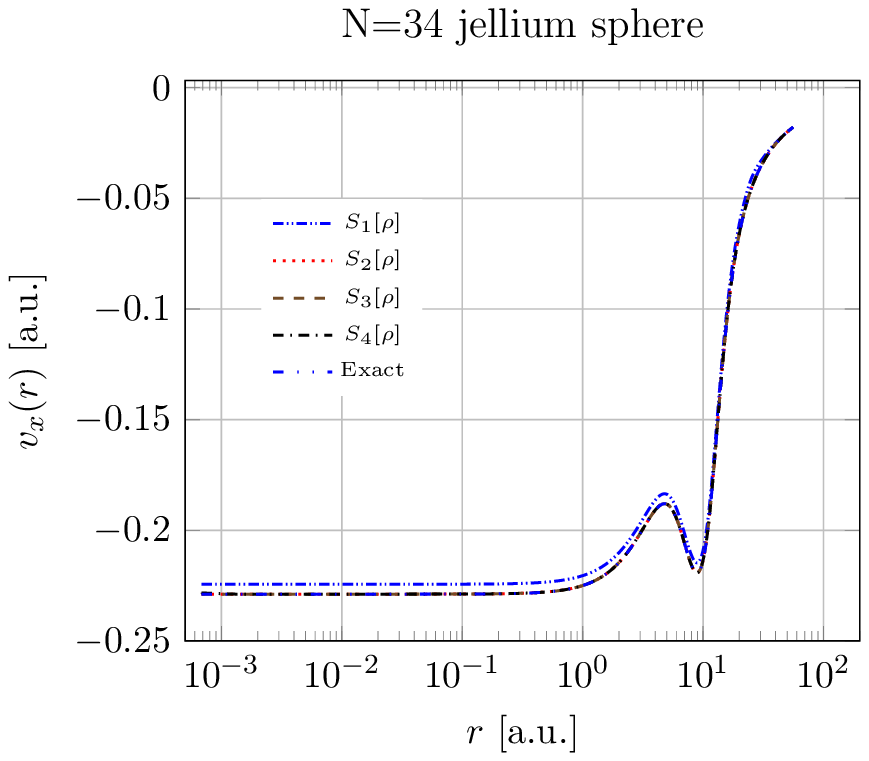}}
	\hfill
	\subfloat[\label{subfig5.5.2d}]{\includegraphics[scale = 0.95]{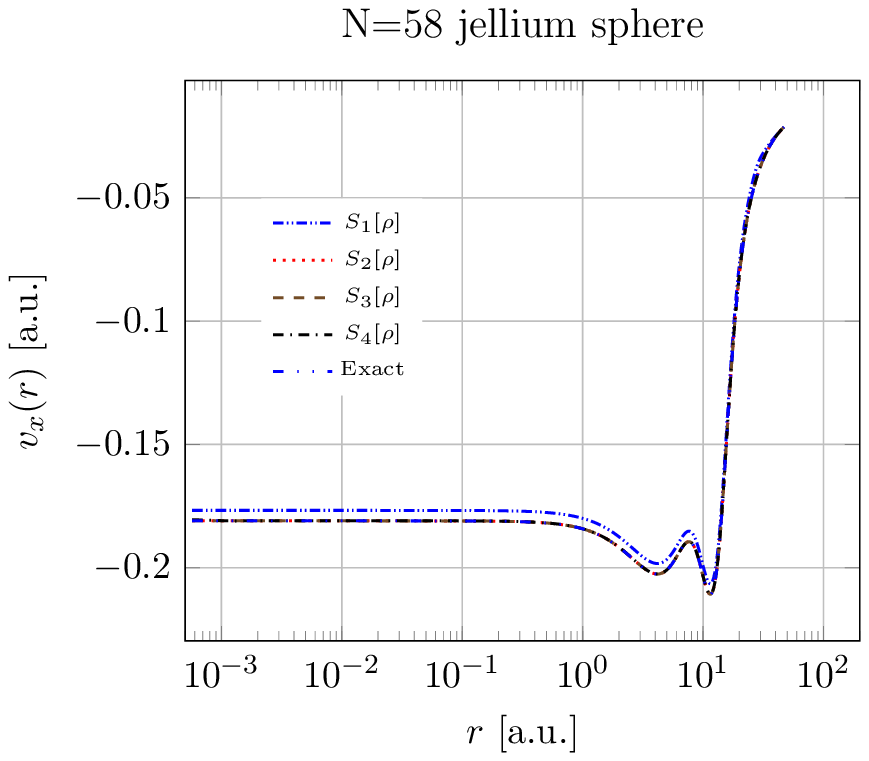}}
	\caption{ \label{Fig5.5.2} Exchange potential for Harbola-Sahni electronic densities of jellium spheres corresponding to number of atoms N=8,18,34,58, with each atom contriduting one electron, obtained by employing functionals $S_1[\rho]$,$S_2[\rho],$$S_3[\rho]$,and $S_4[\rho]$ Eqs. (\ref{eq5.5.1}- \ref{eq5.5.4}).}
\end{figure*}
\begin{table}
	\caption{\label{tab5.5a } Results for Hartree-Fock density of Be, Ne, Na, and Ar atoms  corresponding to functionals $S_1[\rho]$,$S_2[\rho]$,$S_3[\rho]$,and $S_4[\rho]$ defined in Eqs. (\ref{eq5.5.1}- \ref{eq5.5.4}). We have listed the  $\epsilon_{max}$  eigenvalue of the highest occupied Kohn-Sham orbital and kinetic energy $T_S$. In bracket, we also have shown the eigenvalue of the highest occupied HF orbital and HF kinetic energy of every atom.  All the values are in the atomic unit.}
	\begin{ruledtabular}
		\begin{tabular}{lccc}
			& $S[\rho]$          &$\epsilon_{max}$   &$T_S $       \\ \hline
			Be               & $S_1[\rho]$         & -0.3118       &         14.5725  \\
			& $S_2[\rho]$         & -0.3118      &14.5724      \\
			& $S_3[\rho]$   & -0.3118      & 14.5724   \\
			& $S_4[\rho]$   &  -0.3107      &14.5724    \\
			&                       & (-0.3093) & (14.5730)   \\  \\
			Ne      & $S_1[\rho]$   & -0.8451    &  128.5454    \\
			& $S_2[\rho]$    &   -0.8503 &   128.5446  \\
			& $S_3[\rho]$  &    -0.8503 &   128.5448   \\
			& $S_4[\rho]$  &   -0.8468  &   128.5453  \\
			&                       & (-0.8504) &      (128.5471)    \\ \\
			Na    & $S_1[\rho]$    & -0.1822   &  161.8565   \\
			& $S_2[\rho]$    & -0.1821  &  161.8558    \\
			& $S_3[\rho]$    & -0.1821  &161.8559    \\
			& $S_4[\rho]$    & -0.1821  &   161.8565   \\
			&                         & (-0.1821)    & (161.8589)   \\ \\
			Ar    & $S_1[\rho]$   &-0.6066   &526.8124 \\
			& $S_2[\rho]$   &-0.6024 & 526.8110  \\
			& $S_3[\rho]$  & -0.5990  & 526.8124   \\
			& $S_4[\rho]$  & -0.5947&526.8122 \\
			&           &   (-0.5910)  & (526.8175)\\
		\end{tabular}
	\end{ruledtabular}
\end{table}
\begin{table}
	\caption{\label{tab5.5b}Results for the Hookium atom  and  jellium spheres with N=8, 18,34, and 58 atoms.  Density used for Hookium is the exact one and taht for jellium sphere is calculated using Harbola-Shani potential. In the bracket, we have shown the exact chemical potential and exact Kohn-Sham kinetic energy corresponding to the  input densities employed. Caption is same as used in table \ref{tab5.5a }.}
	\begin{ruledtabular}
		\begin{tabular}{lccc c}
			& $S[\rho]$          &$\epsilon_{max}$   &$T_S $       \\ \hline
			Hookium & $S_1[\rho]$      & 1.2514   &     0.6352 \\
			& $S_2[\rho]$       & 1.2497 &       0.6352 \\
			& $S_3[\rho]$       &1.2498   &      0.6352 \\
			& $S_4[\rho]$      & 1.2500  &     0.6352 \\
			&                           & (1.2500)  & ( 0.6352)    \\ \\
			N=8     &$S_1[\rho]$&   -0.1551 &  0.4645\\
			&$S_2[\rho]$&      -0.1582 &  0.4645\\
			&$S_3[\rho]$&             -0.1580&  0.4645\\
			&$S_4[\rho]$&         -0.1582&  0.4645\\
			&                &         (-0.1582)&  (0.4645)\\   \\
			N=18  &$S_1[\rho]$&     -0.1393 & 1.1096\\
			&$S_2[\rho]$&         -0.1427& 1.1096\\
			&$S_3[\rho]$&          -0.1427  &  1.1096\\
			&$S_4[\rho]$&          -0.1427&  1.1096\\
			&                       & (  -0.1427) & (1.1096) \\ \\
			N=34   &$S_1[\rho]$&     -0.1300  & 2.1476\\
			&$S_2[\rho]$&         -0.1345 & 2.1476\\
			&$S_3[\rho]$&           -0.1344 & 2.1476\\
			&$S_4[\rho]$&          -0.1344 & 2.1476\\
			&                   &         ( -0.1344)& (2.1476) \\ \\
			N=58     &$S_1[\rho]$&      -0.1245 &3.6984\\
			&$S_2[\rho]$&         -0.1288& 3.6984\\
			&$S_3[\rho]$&          -0.1288& 3.6984\\
			&$S_4[\rho]$&             -0.1287& 3.6984\\
			&                  & ( -0.1288) & (3.6984)
		\end{tabular}
	\end{ruledtabular}
\end{table}
\begin{figure}
	\centering
	\includegraphics[scale = 0.95]{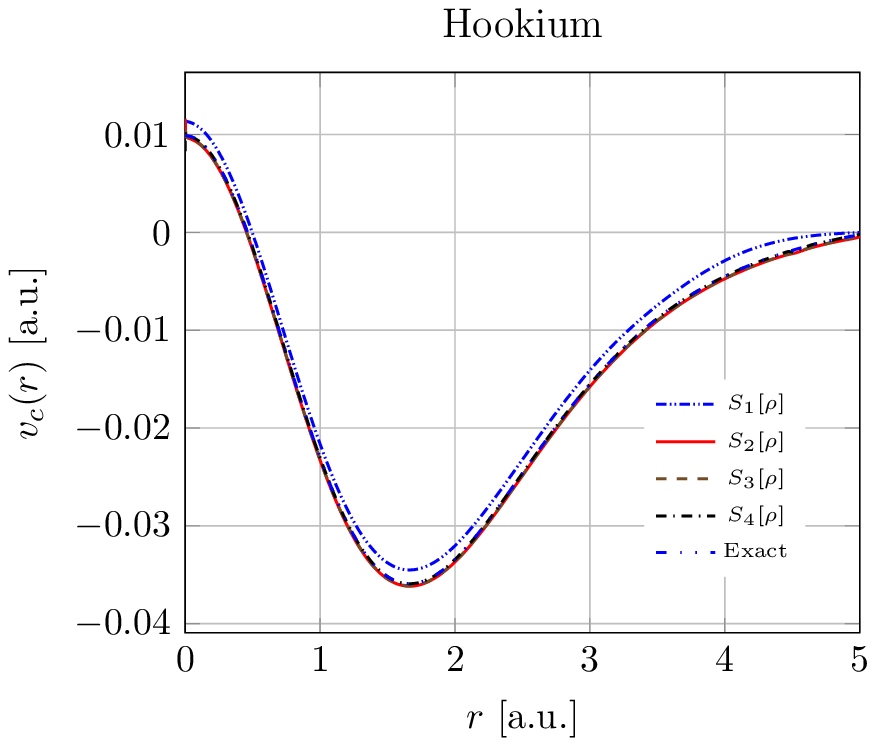}
	\caption{ \label{Fig5.5.3}Correlation  potential of Hookium calculated by employing the   functionals  $S_1[\rho]$,$S_2[\rho]$, $S_3[\rho]$,and $S_4[\rho]$ of Eqs. (\ref{eq5.5.1}- \ref{eq5.5.4}).}
\end{figure}

\section{A general penalty method for obtaining $v_{xc}(\vect{r})$}\label{sec5.4}
Although initial implementation of Eq. (\ref{eq5.2.2.5}) and subsequent work expressed $v_{xc}(\vect{r})$ as a sum of Gaussian functions, it can equally well be done  using an iterative process. For this, one  starts with an approximate $v_{xc}(\vect{r})$ and then updates it using a functional $S[\rho]$ of the dimension of energy such that
\begin{equation}\label{eq5.4.1}
v_{xc}^{i+1}(\vect{r}) = v_{xc}^{i}(\vect{r}) + \frac{\delta S[\rho]}{\delta \rho(\vect{r})}\Big|_{ \rho_i(\vect{r})} - \frac{\delta S[\rho]}{\delta \rho(\vect{r})}\Big|_{ \rho_0(\vect{r})}
\end{equation}
where $i$ indicates the  iteration cycle number. The functional derivative of $S[\rho]$ is required to satisfy the condition \cite{Kumar_2019}
\begin{equation}\label{eq5.4.2}
\int \Big( \frac{\delta S[\rho]}{\delta \rho(\vect{r})}\Big|_{ \rho_i(\vect{r})} - \frac{\delta S[\rho]}{\delta \rho(\vect{r})}\Big|_{ \rho_0(\vect{r})} \Big)(  \rho_i(\vect{r}) -  \rho_0(\vect{r}) ) d\vect{r} \ge 0,
\end{equation}
so that in each iteration the functional given by Eq. (\ref{eq5.2.2.3}) becomes larger and larger reaching ultimately the functional $F[\rho_0]$.
\par One of the choices for the functional $S[\rho]$ is
\begin{equation}\label{eq5.4.3}
S[\rho] =\frac{\epsilon}{2}\iint \frac{\rho(\vect{r})\rho(\vect{r}')}{|\vect{r}-\vect{r}'|}d\vect{r}'d\vect{r}'
\end{equation}
i. e. the Hartree energy functional, with $\epsilon$ being a small number. This gives
\begin{equation}\label{eq5.4.4}
\frac{\delta S[\rho]}{\delta \rho(\vect{r})} = \epsilon\int \frac{\rho(\vect{r}')}{|\vect{r}-\vect{r}'|}d\vect{r}'.
\end{equation}
Thus for this choice of $S[\rho]$, the Hartree potential  updates the exchange-correlation potential in each cycle as given by Eq. (\ref{eq5.4.1}). The condition of convergence satisfied by this $S[\rho]$, as given by  Eq. (\ref{eq5.4.2}), is
\begin{equation}\label{eq5.4.5}
\epsilon \iint \frac{ \{\rho (\vect{r})  -\rho_0(\vect{r}) \}\{\rho (\vect{r}') -\rho_0(\vect{r}')\}}{|\vect{r}-\vect{r}'|}d\vect{r}d\vect{r}' \ge 0.
\end{equation}
Notice that the integral in the equation above is precisely the same as the  penalty functional employed in Eq. (\ref{eq5.2.1.1}) for the implementation of Levy's constrained-search. This suggests that not only the functional given by Eq. (\ref{eq5.2.1.1}) or by Eq. (\ref{eq5.4.5}) but a general functional
\begin{equation}
\int \Big( \frac{\delta S[\rho]}{\delta \rho(\vect{r})} - \frac{\delta S[\rho]}{\delta \rho(\vect{r})}\Big|_{\rho_0(\vect{r})} \Big)(  \rho(\vect{r}) -  \rho_0(\vect{r})) d\vect{r},
\end{equation}
 where $S[\rho]$ is a functional satisfying Eq. (\ref{eq5.4.2}), can be used as a penalty functional for performing Levy's constrained-search. The $S[\rho]$ functionals are precisely those that are employed in updating the exchange-correlation potential iteratively using Eq. (\ref{eq5.4.1}) to find the Kohn-Sham potential for a given density. Then the  equation to be solved for obtaining the exchange-correlation potential is
\begin{equation}\label{eq5.4.6}
\Big[ -\frac{1}{2} \nabla^2 + v_{ext}(\vect{r}) +v_H(\vect{r}) + v^{\lambda}[\rho,\rho_0](\vect{{r}}) \Big ]\phi_i(\vect{r}) = \epsilon_i\phi_i(\vect{r}),
\end{equation}
where
\begin{equation}\label{eq5.4.7}
\begin{split}
&v^{\lambda}[\rho,\rho_0](\vect{{r}}) \\
 &= \lambda \frac{\delta }{\delta \rho(\vect{r})}\Big[\int \Big( \frac{\delta S[\rho]}{\delta \rho(\vect{r}')} - \frac{\delta S[\rho]}{\delta \rho(\vect{r}')}\Big|_{\rho_0(\vect{r}')} \Big)(  \rho(\vect{r}') -  \rho_0(\vect{r}'))d\vect{r}' \Big] \\
&= \lambda \Big[   \frac{\delta S[\rho]}{\delta \rho(\vect{r})} - \frac{\delta S[\rho]}{\delta \rho(\vect{r})}\Big|_{ \rho_0(\vect{r})} \\
& ~~~~~~~~~~~~~~~~~~~+ \int \frac{\delta^2 S[\rho]}{\delta \rho(\vect{r})\delta \rho(\vect{r}')}(\rho(\vect{r}') -  \rho_0(\vect{r}'))d \vect{r}'\Big]
\end{split}
\end{equation}
and the exchange-correlation potential is given as
\begin{equation}\label{eq5.4.8}
v_{xc}(\vect{r})
=\lim_{\lambda \to ~\infty} v^{\lambda}[\rho,\rho_0](\vect{{r}}) .
\end{equation}
With the prescription given in Eq. (\ref{eq5.4.8}),  we now have a general penalty method where a  functional $S[\rho]$ is applied in Levy's constrained minimization method. This functional is the same as  that used for maximization of functional $J_{\rho_0}[v] \cite{Kumar_2019}$.
\par The presentation above brings  to  fore the complimentary nature, as discussed in section \ref{subsec_cojugate}, of the the two ways of obtaining universal functional $F[\rho]$ and unifies them through the functional $S[\rho]$. On the operational side it connects the two method through the functional $S[\rho]$ and extends the Zhao-Morisson-Parr method to a general penalty method to obtain the exchange-correlation potential for a given density.

\section{Results}\label{sec5.5}
We now perform the general penalty method calculations as described in the previous section  using the functionals \cite{Kumar_2019}
\begin{equation} \label{eq5.5.1}
S_1[\rho] = \frac{1}{2}\iint \frac{\rho(\vect{r})\rho(\vect{r}')}{|\vect{r}-\vect{r}'|}d\vect{r}'d\vect{r}' ,
\end{equation}
\begin{equation}\label{eq5.5.2}
S_2[\rho] = \int \rho(\vect{r}) \log(\rho(\vect{r}))d\vect{r} ,
\end{equation}
\begin{equation}\label{eq5.5.3}
S_3[\rho] = \frac{1}{(n+1)}\int \rho^{(n+1)}(\vect{r})d\vect{r}, ~ ( n>0) ,
\end{equation}
and
\begin{equation}\label{eq5.5.4}
S_4[\rho] =  - \frac{1}{2}\int \rho^{1/2}(\vect{r})\nabla^2\rho^{1/2}(\vect{r})d\vect{r}.
\end{equation}
 For  $S_3[\rho]$ we have taken $n$ to be $0.05$. We have calculated the exchange-correlation potential for   electronic  densities of atoms,  jellium spheres and the Hookium atom.  For the atomic systems we have used Hartree-Fock density \cite{Bunge_1993} of Be, Ne, Na, and Ar atoms. For jellium spheres \cite{Knight_PRL.52.2141,Matthias_RMP.65.677} we have employed their electronic densities  corresponding to number of atoms N=8,18,34,58, with each atom contributing one electron. These densities are obtained using the Harbola-Sahni quantal-DFT  method \cite{MKH_89,SAHNI_BOOK}. In case of the Hookium atom exact density \cite{Laufer_1986,Kais_1993} is  employed . The systems considered here have different external potentials. The external potential is proportional to $-\frac{1}{r}$ and  $r^2$, respectively,  for  atoms and the Hookium. For  jellium spheres it depends on $r^2$ inside the  sphere and goes as $-\frac{1}{r}$ outside the  sphere. Here $r$ is the distance from the nucleus or from the cener of the Hookium and the jellium spheres.
\par  To obtain the exchange-correlation potential using Eq. (\ref{eq5.4.8}), one usually solves the Kohn-Sham Eq. (\ref{eq5.4.6}) self-consistently for a series of values of $\lambda$ and obtain the corresponding exchange-correlation potentials. Next, these exchange-correlation potentials are employed to get the exact  exchange-correlation potential through some extrapolation technique. However,  we use an alternate iterative approach as suggested in ref \cite{morrison_1995}. We start with a small value  $\lambda$ say $\lambda_j$ $(j =0)$ and  self-consistently solve  Eq. (\ref{eq5.4.6}) for the potential
\begin{equation}  \label{eq5.5.5}
v_{KS}(\vect{r})= v_{fixed}(\vect{r})  + v^{\lambda_j}[\rho_0,\rho](\vect{r}).
\end{equation}
 Here
\begin{equation}\label{eq5.5.6}
v_{fixed}(\vect{r}) = v_{ext}(\vect{r})+(1-\frac{1}{N})v_H[\rho_0](\vect{r})
\end{equation}
is part of the  potential that is known exactly and is kept fixed, and
\begin{small}
\begin{equation}\label{eq5.5.7}
\begin{split}
v^{\lambda_j}[\rho_0,\rho](\vect{r})= \lambda_j  \Big[   &\frac{\delta S[\rho]}{\delta \rho(\vect{r})} - \frac{\delta S[\rho]}{\delta \rho(\vect{r})}\Big|_{ \rho_0(\vect{r})} \\
& + \int \frac{\delta^2 S[\rho]}{\delta \rho(\vect{r})\delta \rho(\vect{r}')}(\rho(\vect{r}') -  \rho_0(\vect{r}'))d \vect{r}'\Big]
\end{split}
\end{equation}
\end{small}
is the update term during the self-consistent process. In Eq. (\ref{eq5.5.6}), N is the total numbers of electrons. The Hartree term is multiplied by the factor $(1-\frac{1}{N})$ to make Fermi-Amaldi correction \cite{fermi1934orbite} so that the self-interaction component of the exchange potential is accounted for exactly; This facilitates obtaining the exchange-correlation potential by making the self-consistent (SCF) calculation less demanding. Having obtained the self-consistent  solution for $\lambda_j$, next we increase the value of $\lambda$ from $\lambda_j$ to $\lambda_{j+1}$ and again solve the Eq. (\ref{eq5.5.5}) incorporating the SCF potential $ v^{\lambda_j}[\rho_0,\rho](\vect{r})$ with the exactly known part $v_{fixed}(\vect{r})$, and using this as the starting potential for $\lambda _{j+1}$. This process is iterated until the quantity
$ \Delta_{\rho}= \int |\rho_{0}(\vect{r})-\rho(\vect{r})| d\vect{r} $ becomes smaller than some chosen value $\delta_{itr}$. In our calculations we have taken $\delta_{itr}$  to be $1 \times 10^{-5}$ for atoms and   $ 1 \times 10^{-6}$ for jellium spheres and Hookium atom. During the self-consistent cycle a linear mixing of density has been employed for the calculation of potential. Since the systems we consider are spherically symetric, in the following we write all relevant quantities as a function of radial coordinate $r$.  To check the convergence of  self-consistent process we have taken
\begin{equation}
max\Big|r[ v^{(i)}_{KS}({r})-v^{(i+1)}_{KS}({r})]\Big| \le 1 \times 10^{-8},
\end{equation}
  where $v^{(i)}_{KS}({r})$ and $v^{(i+1)}_{KS}({r})$ are the  potentials of two consecutive cycles.  For all the calculations  we have used a modified Hermann-Skillman code \cite{HSK}.  We have also fixed the potential $v_{KS}(\vect{r})$   to   $ -\frac{1}{r}$ \cite{UvBarth_1985} at a suitably chosen large distance $r$ in     the  asymptotic region.

\par The exchange-correlation potentials   obtained  for different systems by employing  functionals $S[\rho]$ of Eqs. (\ref{eq5.5.1}-\ref{eq5.5.4}) are shown in the Fig (\ref{Fig5.5.1}) ,Fig (\ref{Fig5.5.2}), and Fig (\ref{Fig5.5.3}). In Fig (\ref{Fig5.5.1}) we display our results for atoms. Since we have used the Hartree-Fock density of atoms, potentials obatained by us on inversion of this density should be very close to the exact exchange potentials for these systems. Therefore, we compare our results with corresponding exact exchange potentials given by the optimized potential method  (OPM) \cite{PhysRev.90.317,Talman_1978,Engel_1993}.   It is evident from the figure that for the functionals $S_1[\rho]$, $S_2[\rho]$ and $S_3[\rho]$ the  output exchange potential  matches perfectly with the  corresponding exact result. The resulting potentials obtained by functional $S_4[\rho]$ are also  on the top of  exact result except  in the regions very  close to the nucleus.
\par Next in  Fig (\ref{Fig5.5.2}), we have plotted the exchange  potential for jellium spheres that is calculated by inverting the density obatined from the Kohn-Sham calculations employing the Harbola-Shani exchange potential. The potentials calculated  by $S_2[\rho]$, $S_3[\rho]$ and $S_4[\rho]$ are on the top of  the corresponding exact Harbola-Sahni potentials.  On the other hand, for the functional $S_1[\rho]$ potentials   have a small constant shift from the exact results even when its asymptotic values is fixed to $-\frac{1}{r}$. Finally in Fig (\ref{Fig5.5.3}), we have displayed the correlation potential of the Hookium atom along with the exact correlation potential. For this we have used the exact density for the Hookium atom with $\omega  = 0.5$. Again the potentials calculated  by $S_2[\rho]$, $S_3[\rho]$ and $S_4[\rho]$  match with the exact result and the potentials corresponding to functional $S_1[\rho]$  shows a small constant shift.
\par The eigenvalues $\epsilon_{max}$ of  highest occupied  Kohn-Sham orbital    and the Kohn-Sham kinetic energies $T_S[\rho]$ corresponding to functionals $S_1[\rho]$, $S_2[\rho]$, $S_3[\rho]$, and $S_4[\rho]$ are displayed in  Table (\ref{tab5.5a })  for the atoms and in Table (\ref{tab5.5b}) for the Hookium atom and jellium spheres.   It is evident from Table (\ref{tab5.5a }) that for every atom and  for each functional employed,  $\epsilon_{max}$ are close to each other and also close to the  eigenvalue of highest occupied Hartree-Fock orbital. Similarly  $T_S[\rho]$ obtained by different functionals are close to each other and smaller than the corresponding   Hartree-Fock kinetic energy. Furthermore, for the Hookium atom and jellium spheres  Table (\ref{tab5.5b})  shows that  for every functional employed by us,  calculated values of  $T_S[\rho]$  match with  the corresponding exact results. Similarly  $\epsilon_{max}$ also match with the exact results for the functionals $S_2[\rho]$, $S_3[\rho]$ $S_4[\rho]$. However, in the case of functional $S_1[\rho]$ the value of $\epsilon_{max}$ is different from the exact one. This difference is the same as the shift in the potential for $S_1[\rho]$.
\par We point out that application of the functional $S_2[\rho]$ to jellium spheres and the Hookium atom  causes some problem in obtaining the SCF solution of Kohn-Sham equation. This problem  arises because the densities of jellium spheres and the Hookium atom are very small all over the space.  Thus  the  potential calculated by  $S_2[\rho]$ during the self-consistent cycle becomes very large and leads to difficulty in solving the corresponding Kohn-Sham equation.  To overcome this  difficulty, for first few values of $ \lambda$  we mixed both the   potential and the  density  in the SCF calculation.  For a given value of $\lambda_{j}$, take a trial input desity $\rho^{(i)}_{in}(\vect{r})$ and use it in Eq. (\ref{eq5.5.5}) to calculate the potential 
\begin{equation*}
v^{(i)}_{KS}(\vect{r}) = v_{fixed}(\vect{r})  + v^{\lambda_j}[\rho_0,\rho^{(i)}_{in}(\vect{r})](\vect{r}).
\end{equation*}
Then solve the KS equation (Eq. (\ref{eq5.1.1})) using the potential $v^{(i)}_{KS}(\vect{r})$ and obtain output density $\rho^{(i)}_{out}(\vect{r})$. Then potential for next iteration is  taken as 
\begin{equation*}
(1-\eta)v^{(i)}_{KS} (\vect{r}) + \eta v^{(i+1)}_{KS} (\vect{r}),
\end{equation*}
where the potential $v^{(i+1)}_{KS} (\vect{r})$ is calculated by employing the input density
\begin{equation*}
\rho^{(i+1)}_{in}(\vect{r}) = (1-\epsilon)\rho^{(i)}_{in}(\vect{r}) + \epsilon\rho^{(i)}_{out}(\vect{r}).
\end{equation*}
Here $\epsilon$ and $\eta$ are  mixing parameter.  In our calculations the potential mixing given in equation above is used for the $\lambda$ up to $100$ with  $\eta =0.01$.  Furthermore, we also note that, while using functional $S_4[\rho]$ for  Be and Ar,  $\Delta_{\rho}$ achieves minimum value of  $5 \times 10^{-5}$. Similarly for the functional $S_1[\rho]$ a minimum value of  $2 \times 10^{-6}$ be could achieved in the case of  jellium spheres having  34 and 54 atoms.
\section{Conclusion}
To conclude, in the present work  we have  derived a general  penalty method for  Levy's constrained-search for the universal functional of DFT. This gives the  Kohn-Sham potential and  Kohn-Sham kinetic energy for a given density using several different functionals $S[\rho]$. These functionals are the same as those used in density-to-potential inversion through Lieb's formulation. This brings forth the complementary nature of  Levy's and Lieb's  definition for universal functionals for a given density and  enables us to  generalize the ZMP method using the functional $S[\rho]$.
\par Utility of the present work along that of ref. \cite{Kumar_2019} lies in it giving several methods for obtaining the exchange-correlation potential for a ground-state density $\rho_0(\vect{r})$. Depending on the system  and corresponding density, one can thus choose an appropriate functional for the carrying out these calculation. 
\begin{acknowledgments}
We are grateful to  Prof. Dr. Eberhard Engel for providing optimized effective potential data of atoms.
\end{acknowledgments}
\appendix
\section{Penalty method approach for Levy's constrained search} \label{sec5a1}
Employing Levy's constrained search, one obtains the universal functional $F[\rho]$ of density functional theory through constrained minimization as given by Eq. (\ref{eq5.2.2}). In the context of Kohn-Sham system, this constraint has been enforced using a penalty functional of the form given by Eq. (\ref{eq5.2.1.1}). In this section, we present general theorems of  penalty method for Levy's constrained minimization. These are based on the discussion in ref. \cite{gunaratne2006penalty}.
\par Consider the functional $F[\Psi] =\langle \Psi |T +V_{ee}| \Psi\rangle$ of wavefunction $\Psi$   to be minimized with the constraint that the density  corresponding to $\Psi$ is equal to the given ground-state density $\rho_0(\vect{r})$. note that $F[\Psi] \geq 0 $. Let  $F[\rho_0] $ be the minimum value this functional. Then in  Levy's constrained-search method \cite{Levy_1979} this is obtained by
\begin{equation} \label{apn_levy}
F[\rho_0] = \underset{\Psi \to \rho_0(\vect{r})}{min} F[\Psi].
\end{equation}
In writing Eq. (\ref{apn_levy}), it is assumed  \cite{Lieb_1983} that $F[\Psi]$ is continuous and has a minimum. To approach this problem through penalty method we introduce a  functional of $\Psi$
\begin{equation}\label{eq5a1.1}
F_P[\Psi,\lambda] = F[\Psi] +\lambda P[\rho,\rho_0],
\end{equation}
where $\rho(\vect{r})$ is the density corresponding to $\Psi$ and  $P[\rho, \rho_0]$  is  a penalty functional.  It is chosen such that  $P[\rho, \rho_0] \ge 0$ where equality is satisfied for $\rho (\vect{r}) = \rho_0(\vect{r})$. The quantity $\lambda > 0$,  is known as the penalty parameter.  Then   minimizing $F[\Psi]$  with density constraint is  equivalent to $\lim\limits_{\lambda \to \infty} {min}~F_P[\Psi,\lambda]$. Thus constrained minimization is mapped to an unconstrained minimization. In this minimization, as $\lambda \to \infty$,  $\rho(\vect{r})$ corresponding to  $\Psi$ approaches $\rho_0(\vect{r})$ and the  functional $F[\Psi]$  obtains the minimum value corresponding to $\rho_0(\vect{r})$. This is shown in the following.
\par Let the $\lambda_k > 0,~ k= 1,2,....,\infty$  be a sequence of penalty parameters such that $\lambda_{k+1} >\lambda_k$ and $\lim\limits_{\lambda \to \infty} \lambda_{k} = \infty$ and  let $\Psi_k$ be the minimizing function  for the functional $F_P[\Psi,\lambda_k]$; We assume that the sequence $\Psi_k$ is convergent. Furthermore, let $\Psi^*$ be the  minimizing function of  the universal functional $F[\Psi]$ under the given constraint that $\rho(\vect{r}) = \rho_0(\vect{r})$ where $\rho(\vect{r})$ is the density for  $\Psi^*$. Then the  following relations hold:\\
\begin{rel} \label{eq5a1.2}
$F_P[\Psi_{k+1},\lambda_{k+1}] \ge F_P[\Psi_{k},\lambda_{k}]$
\end{rel}
\begin{proof} $\Psi_{k+1}$ and  $\Psi_{k}$ are the solutions of functional $F_P[\Psi,\lambda]$ for $\lambda$ equal to  $\lambda_{k+1}$ and $\lambda_{k}$, respectively, where $\lambda_{k+1} > \lambda_{k}$. Let $\rho_{k+1}(\vect{r})$ and $\rho_{k}(\vect{r})$ be the densities calculated from wavefunctions $\Psi_{k+1}$ and $\Psi_{k}$, respectively. Then we have
\begin{equation}
\begin{split}
F_P[\Psi_{k+1},\lambda_{k+1}] &= F[\Psi_{k+1}] +\lambda_{k+1}P[\rho_{k+1},\rho_0]\\
&>  F[\Psi_{k+1}] +\lambda_{k}P[\rho_{k+1},\rho_0] ~ ( \text{becuase~}  \lambda_{k+1} >\lambda_k)\\
&    \ge  F_P[\Psi_{k},\lambda_{k}]
\end{split}
\end{equation}
\end{proof}
The last inequality above follows because $F_P[\Psi_{k},\lambda_{k}]$ has minimum value for   $\lambda_{k}$. Thus as $k$ increases, value of the corresponding functional $F_P[\Psi_{k},\lambda_{k}]$ also increases.
\begin{rel} \label{eq5a1.3}
Now  we show that $P[\rho_{k+1},\rho_0] \le P[\rho_k,\rho_0].$
\end{rel}
\begin{proof} Start with
\begin{equation}\label{eq5a1.4}
F_P[\Psi_{k+1},\lambda_{k+1}] \le  F_P[\Psi_{k},\lambda_{k+1}]
\end{equation}
and
\begin{equation}\label{eq5a1.5}
F_P[\Psi_{k+1},\lambda_{k}] \ge  F_P[\Psi_{k},\lambda_{k}].
\end{equation}
Now by subtracting   Eq. (\ref{eq5a1.5})  from Eq. (\ref{eq5a1.4}) we get
\begin{equation}
\begin{split}
&F_P[\Psi_{k+1},\lambda_{k+1}] - F_P[\Psi_{k+1},\lambda_{k}]\le    F_P[\Psi_{k},\lambda_{k+1}] -F_P[\Psi_{k},\lambda_{k}] \\
&    \implies    P[\rho_{k+1},\rho_0]  \le P[\rho_{k},\rho_0]
\end{split}
\end{equation}
\end{proof}
\begin{rel} \label{eq5a1.6}
Using the two  relations above, we now show that the functional $F[\Psi]$ too increases as  $k$ increases in the sequence i.e.
$F[\Psi_{k+1}] \ge F[\Psi_{k}] $
\end{rel}
    \begin{proof}  From Relation (\ref{eq5a1.2}) and  Relation (\ref{eq5a1.3}) we have
\begin{equation}\label{eq5a1.6a}
F_P[\Psi_{k+1},\lambda_{k}] \ge  F_P[\Psi_{k},\lambda_{k}]
\end{equation}
and
\begin{equation}
\lambda_k P[\rho_{k+1},\rho_0]  \le  \lambda_k P[\rho_{k},\rho_0]
\end{equation}
or
\begin{equation}\label{eq5a1.6b}
-\lambda_k P[\rho_{k+1},\rho_0] \ge  -\lambda_k P[\rho_{k},\rho_0]
\end{equation}
Adding the   Eq. (\ref{eq5a1.6a})  and Eq. (\ref{eq5a1.6b}) gives
\begin{equation}
F_P[\Psi_{k+1},\lambda_{k}]- \lambda_k P[\rho_{k+1},\rho_0]\ge  F_P[\Psi_{k},\lambda_{k}]
 -\lambda_k P[\rho_{k},\rho_0].
\end{equation}
$    \implies $ $F[\Psi_{k+1}] \ge   F[\Psi_{k}] $.
\end{proof}
Notice that now we have shown that as the penalty parameter $\lambda$ increases with $k$, the functional $F[\Psi_k]$  keeps becoming larger and larger. On the other hand, the penalty functional $P[\rho,\rho_0]$ is not yet zero. However, as shown in the theorem below, $P[\rho,\rho_0] \to 0$ when $\lambda \to \infty$.
\begin{rel} \label{lem4}
 $  F[\Psi_{k}] \le F[\Psi^*] $
\end{rel}
\begin{proof} Note that  $\Psi^*$  satisfies the density constraint. Thus penalty functional  $P[\rho,\rho_0] = 0 $ for the density corresponding to $\Psi^*$ . It follows that
\begin{equation}
\begin{split}
F[\Psi_k] &\le F[\Psi_k] + \lambda_k P[\rho_k,\rho_0]  ~ (\text {because ~} P[\rho_k,\rho_0] \ge 0)\\
             &= F_P[\Psi_k, \lambda_k]  \\
                & \le  F_P[\Psi^*, \lambda_k] =F[\Psi^*].
\end{split}
\end{equation}
\end{proof}
All the relations above lead to the final result that is given in the theorem below.
\begin{theorem}
Let the $\{\Psi_k\}, ~ k= 1,2,....,\infty $ be a convergent sequence of functions that minimize $F_P[\Psi,\lambda_{k}]$ with  $\lambda_{k+1} >\lambda_k$. Then the  limit $\bar{\Psi}$ of the set $\{\Psi_k\}$ gives the  minimum of $F[\Psi]$ satisfying the  constraint $P[\rho, \rho_0] = 0$.
\end{theorem}
\begin{proof} It is given that $\Psi^*$ gives minimum of  functional $F[\Psi]$ under the given constraint. By applying  Relation (\ref{lem4})  for $k \to \infty$ we have
\begin{equation}
F[\bar{\Psi}] = \lim\limits_{k \to \infty}F[\Psi_k] \le  \lim\limits_{k \to \infty}F_P[\Psi_k, \lambda_k]  \le  F[\Psi^*],
\end{equation}
where the last inequality above follows from the proof of Relation (\ref{lem4}).  This shows that $\lim\limits_{k \to \infty}F_P[\Psi_k, \lambda_k]$ is bounded.  Hence
$\lim\limits_{k \to \infty} \lambda_{k} P[\rho_{k}, \rho_0]$ is finite. Since for $k \to \infty$ the penalty parameter $\lambda_k \to \infty$, it follows that $\lim\limits_{k \to \infty} P[\rho_{k}, \rho_0]\to 0 $. The inequality above thus shows that : (i) $P[\rho_{k}, \rho_0] = 0$ for $k \to \infty$ ; (ii) and therefore $F[\bar{\Psi}]$ is the  optimal solution of  $F[\Psi]$ with the given constraint. Thus we have  $F[\bar{\Psi}] \le F[\Psi^*]$. Now by assumption $\Psi^*$ is the solution for the minimum of $F[\Psi]$ so  $F[\bar{\Psi}] \ge F[\Psi^*]$. Therefore  $F[\bar{\Psi}] = F[\Psi^*]$ and we conclude that limit $\bar{\Psi}$ of set $\{\Psi_k\}$ is  the  minimum of $F[\Psi]$ with $\lim\limits_{k \to \infty} P[\rho_{k}, \rho_0] = 0$.
\end{proof}
\bibliographystyle{apsrev4-1}
\bibliography{ak_thesis.bib}
\end{document}